\documentclass[pdflatex,sn-mathphys-num]{sn-jnl}

\usepackage{graphicx}%
\usepackage{multirow}%
\usepackage{amsmath,amssymb,amsfonts}%
\usepackage{amsthm}%
\usepackage{mathrsfs}%
\usepackage[title]{appendix}%
\usepackage{xcolor}%
\usepackage{textcomp}%
\usepackage{manyfoot}%
\usepackage{booktabs}%
\usepackage{algorithm}%
\usepackage{algorithmicx}%
\usepackage{algpseudocode}%
\usepackage{listings}%

\theoremstyle{thmstyleone}%
\newtheorem{theorem}{Theorem}
\theoremstyle{thmstyletwo}%
\newtheorem{lemma}{Lemma}
\theoremstyle{thmstylethree}%

\begin{document}

\title[]{On the four-body limaçon choreography: maximal superintegrability and choreographic fragmentation}


\author*[1]{\fnm{} \sur{Adrian M Escobar-Ruiz}}\email{admau@xanum.uam.mx}
\equalcont{These authors contributed equally to this work.}

\author[1]{\fnm{} \sur{Manuel Fernandez-Guasti}}\email{mfg@xanum.uam.mx}
\equalcont{These authors contributed equally to this work.}

\affil*[1]{\orgdiv{Departamento de Física}, \orgname{Universidad Autónoma Metropolitana - Iztapalapa}, \orgaddress{\street{Rafael
Atlixco 186}, \city{CDMX}, \postcode{09340}, \state{Ciudad de México}, \country{México}}}


\abstract{In this paper, as a continuation of [Fernandez-Guasti,  \textit{Celest Mech Dyn Astron} 137, 4 (2025)], we demonstrate the maximal superintegrability of the reduced Hamiltonian, which governs the four-body choreographic planar motion along the trisectrix limaçon (resembling a folded figure eight), in the six-dimensional space of relative motion. {The pairwise interaction potential $V(r_{ij})$ among the four bodies is a quadratic expression in the relative distances $r_{ij}$, with a combination of positive and negative coefficients}. The corresponding eleven integrals of motion in the Liouville-Arnold sense are presented explicitly. Specifically, it is shown that the reduced Hamiltonian admits complete separation of variables in Jacobi-like variables. The emergence of this choreography is not a direct consequence of maximal superintegrability. Rather, it originates from the existence of \textit{particular integrals} and the phenomenon of \textit{particular involution}. We also provide a detailed analysis of the fragmentation of a general four-body choreographic motion into two isomorphic two-body choreographies, as well as the reverse process, namely, the fusion of two-body choreographies into a four-body configuration. This model combines choreographic motion with maximal superintegrability, a seldom-studied interplay in classical mechanics.}





\maketitle

\section{Introduction}\label{sec1}

The classical $n$-body problem \cite{lectures_siegel_1971, solar_murray_1999} has long stood at the heart of celestial mechanics and mathematical physics, offering profound insights into the collective dynamics of interacting particles. This problem, which seeks to determine the motion of $n$ point masses under mutual interactions (typically Newtonian gravity), has served as a testbed for many fundamental ideas in dynamical systems, symmetry, and integrability. 
While the two-body problem is maximally superintegrable, it admits the maximum number of independent integrals of motion (three, in the planar case), and is solvable in closed form\footnote{Its solution can be expressed in closed form in terms of a suitable angular parameter (e.g., the true anomaly). As a function of time, closed-form solutions exist only in specific cases such as circular or rectilinear motion.}, the three-body problem already displays a rich tapestry of chaotic, resonant, and periodic behaviors. The four-body problem \cite{relative_sim_1978} remains an active area of research \cite{ restricted_scheeres_1998, finiteness_hampton_2006,transport_alvarezramrez_2015, classical_escobarruiz_2022}, where both analytical and numerical approaches continue to yield new insights.

\vspace{0.2cm}

One particularly fascinating class of solutions within the \( n \)-body framework is that of \textit{choreographies} \cite{davies1983classical,Chenciner2002}. These are periodic solutions where all bodies move along the same trajectory, with an equal phase delay in time. Since the discovery of the celebrated figure-eight solution in the planar three-body Newtonian problem \cite{braids_moore_1993, remarkable_chenciner_2000}, choreographic motions have attracted significant attention due to their mathematical elegance, symmetry properties, and potential applications in fields ranging from celestial mechanics to atomic and molecular dynamics.

\vspace{0.2cm}

Choreographic solutions typically emerge in systems with high symmetry and interactions depending only on relative distances~\cite{braids_moore_1993, remarkable_chenciner_2000, Toshiaki, Ozaki2009, Kapela2007, LOPEZVIEYRA20191711}. Most known examples have been obtained numerically through variational~\cite{chenciner2003action} or topological~\cite{Montgomery2002InfinitelyMS} techniques, while only a few allow exact analytic solutions. One such case is the planar three-body choreography discovered in~\cite{Toshiaki}, where the particles in a modified Newtonian gravity potential move along a lemniscate of Bernoulli; this lemniscate resembles the shape of the figure-eight Newtonian choreography. However, it remains unclear whether this system, characterized by an attractive 3-body Newtonian potential plus repulsive quadratic pairwise interactions, is (super)integrable in the Liouville–Arnold sense~\cite{Arnold:1989}. Later, the work \cite{TV2020} revealed that certain nontrivial \textit{integrals} of motion (distinct from the classical Liouville-Arnold integrals) can govern these analytical choreographies (see also \cite{TV2021}). These \textit{particular integrals} \cite{Escobar-Ruiz2024,Turbiner_2013} do not generate continuous symmetries. In contrast to generic integrability, which typically relies on global separation of variables, the existence of particular integrals enables the emergence of exact analytic solutions even in non-separable or chaotic systems \cite{Escobar-Ruiz2024}.

\vspace{0.2cm}

Another analytical example, central to our analysis, is the planar four-body choreography along a trisectrix lima\c{c}on curve~\cite{Fernandez-Guasti2025}, where each particle follows the same closed trajectory with a fixed time lag, and the entire configuration exhibits both spatial and temporal symmetry. In this case, the pairwise interactions are quadratic in the relative distances between particles and include both attractive and repulsive forces.

\vspace{0.2cm}

The present study examines the integrability properties of the four-body Hamiltonian introduced in~\cite{Fernandez-Guasti2025}, which governs the motion along the trisectrix limaçon. Our goal is to elucidate the interplay between choreographic motion and superintegrability within a concrete, analytically tractable model. We emphasize that, although superintegrable systems such as the Calogero–Moser model~\cite{Calogero, MOSER1975197} exhibit periodic behavior reminiscent of choreographies, they often fail to satisfy the strict condition of identical trajectories with equal phase shifts. Even the two-body Kepler problem—despite being maximally superintegrable—does not support choreographic motion. This highlights the relevance of the present example, where such motion occurs in a system that is not only analytically solvable but also maximally superintegrable.

\vspace{0.2cm}

The main contributions of this work are threefold. First, we show that the reduced Hamiltonian—after separating the center-of-mass motion—is maximally superintegrable in the Liouville-Arnold sense~\cite{Arnold:1989}, possessing eleven independent integrals of motion in a twelve-dimensional phase space. This structure guarantees that all bounded trajectories are closed~\cite{nekhoroshev1972action}. Second, we demonstrate that the limaçon choreography is sustained by particular integrals and the phenomenon of \textit{particular involution}, in which some global integrals Poisson-commute only along the special solution. Third, we describe a fragmentation mechanism wherein the four-body choreography decomposes into two separate two-body choreographies when these particular integrals are perturbed and cease to be conserved. The inverse process, involving the fusion of two-body choreographies into a four-body configuration, is also examined.

\vspace{0.2cm}

The paper is organized as follows. Section~\ref{sec2} introduces the Hamiltonian system and establishes notation, presenting general analytic solutions and a classification of trajectories. In Section~\ref{sec3}, we perform a full separation of variables and reveal the anisotropic oscillator structure of the reduced Hamiltonian ${\cal H}_{\rm rel}$. Section~\ref{relmotionsup} establishes its maximal superintegrability and constructs eleven algebraically independent integrals (six of them in involution). Section~\ref{sec4} focuses on the limaçon choreography and its relation to particular integrals and particular involution. Section~\ref{sec5} analyzes the fragmentation into two two-body choreographies. Conclusions and future directions are given in Section~\ref{sec6}. 

\section{Model Definition and Setup}\label{sec2}

Consider a classical system of four particles, moving in the Euclidean plane $\mathbb{R}^2$, with equal masses ($m_1=m_2=m_3=m_4=m$) and subject to a quadratic pairwise interaction potential. The corresponding Hamiltonian is of the form:
\begin{equation}
\label{Ha}
{\cal H} \ = \ \frac{1}{2\,m}(\, \mathbf{p}_1^2 \,+\,\mathbf{p}_2^2 \,+\,\mathbf{p}_3^2 \,+\,\mathbf{p}_4^2 \,) \ + \ V(r_{ij}) \ ,
\end{equation}
$\mathbf {p}_i$ is the canonical momentum of particle $i=1,2,3,4$, the potential is
\begin{equation}
\label{Vp}
    V(r_{ij}) \ = \ \frac{1}{4}\,m\,\omega^2\,(\, 2\,r_{12}^2 \,+\, 2\,r_{23}^2 \,+\, 2\,r_{34}^2\,+\, 2\,r_{14}^2\,-\, r_{13}^2 \,-\, r_{24}^2 \,) \ ,
\end{equation}
here $\omega$ is an angular frequency, and
\[ \mathbf{r}_{ij}\ = \  \mathbf{r}_i\,-\,\mathbf{r}_j\ ,
\]
$\mathbf{r}_i \in \mathbb{R}^2$ denotes the individual vector position of particle $i=1,2,3,4$, thus $r_{ij}=|\mathbf{r}_{ij}|$ is the mutual relative distance between bodies $i$ and $j$. The phase space is 16-dimensional. In (\ref{Vp}), the interaction between adjacent particles is attractive with equal strength $\frac{1}{2}\,m\,\omega^2$, and the interaction between alternating particles is repulsive with half that strength.

The potential can be seen as arising from the sum of two linear force contributions, an attractive force between all four bodies and a linear force that is repulsive between alternate numbered bodies and attractive between adjacent ones. This latter force may be viewed as an electrostatic type force but linearly dependent on distance, where odd numbered bodies have the same charge sign but opposite to the even numbered bodies. Further details can be found in~\cite{Fernandez-Guasti2025}.

The Hamiltonian (\ref{Ha}) is rotationally invariant, i.e., it exhibits $SO(2)$ symmetry associated with the conservation of the total angular momentum ${\cal L}= \sum_{i=1}^{4} {\mathbf r}_i \times {\mathbf p}_i$. Because of the spatial translational invariance, the center-of-mass coordinate can be completely separated. This system (\ref{Ha}) possesses a discrete permutation symmetry ${\cal S}_2(13)\otimes {\cal S}_2(24)$, the particle pairs $1,3$ and $2,4$ play a distinctive role in the system's dynamics, and formally the full reflection symmetry $\mathbb{Z}_2^6$ $( r_{ij} \rightarrow - r_{ij})$.  Additionally, it exhibits time-reversal invariance.

In~\cite{Fernandez-Guasti2025} it was shown that (\ref{Ha}) admits an analytic four-body choreographic solution along a trisectrix limax, a curve resembling a folded figure eight. Here we focus on the integrability properties of (\ref{Ha}).

\subsection{General solutions to the equations of motion}

{The Hamiltonian with potential (\ref{Vp}) written in terms of the individual particle positions \(\mathbf{r}_{i}\) does not admit a separation of variables in this form. Nevertheless, Hamilton's equations of motion remain integrable directly in the \(\mathbf{r}_{i}\) coordinates. This is so because the interaction potential is a purely quadratic function of the relative distances (with no linear terms). Thus, by invoking the Principal Axis Theorem \cite{horn2012matrix} the system can be diagonalized into independent normal modes \cite{Landau1976Mechanics}, leading to a decoupled formulation of the dynamics. This normal-mode representation, which provides additional insight into the collective motion of the system, will be described in detail in Section~\ref{sec3}. Hence, the equations
} 
\begin{equation}
\label{eqm}
  \frac{d\mathbf{r}_i }{dt}\ = \ \frac{\partial{\cal H}}{\partial\mathbf{p}_i}   \qquad ; \qquad \frac{d\mathbf{p}_i }{dt} \ = \ -\,\frac{\partial{\cal H}}{\partial\mathbf{r}_i} \qquad \ ; \qquad (i=1,2,3,4)\ ,
\end{equation}
can be solved exactly in terms of elementary trigonometric functions. In Cartesian coordinates, $\mathbf{p}_{i}(t)=m\,\frac{d\mathbf{r}_{i}}{dt}$. For future reference, we introduce the notation $\mathbf{p}_{ij}(t) \equiv \mathbf{p}_{i}(t) - \mathbf{p}_{j}(t)$ to denote the relative momentum between particles $i$ and $j$.

\subsubsection{Classification of trajectories}

The general solutions to (\ref{eqm}) are described in the following Theorem.

\vspace{-0.5cm}

\begin{theorem}
In the center-of-mass frame, the most general motion of the system (\ref{Ha}), with the linear forces allowing both attraction and repulsion, is given by two isomorphic 2-body choreographies.
\end{theorem}

\vspace{-0.5cm}

\begin{proof}

In the center-of-mass frame where
\begin{equation}
\mathbf{r}_1+\mathbf{r}_2+\mathbf{r}_3+\mathbf{r}_4\, = \, 0 \quad \text{and} \qquad  \mathbf{p}_1+\mathbf{p}_2+\mathbf{p}_3+\mathbf{p}_4\, = \, 0 \ ,    
\end{equation}
by direct integration of~(\ref{eqm})
the general solutions, in Cartesian coordinates, are 
\begin{equation}
\begin{aligned}
\label{trajec}
& \mathbf{r}_1(t)    =   \frac{1}{2}\,\big( \,\mathbf{r}_{13}(0)\,\cos{\omega\,t}   +  \mathbf{r}_{13}^+(0)\,\cos{2\,\omega\,t}\,\big)
 \, + \,   \frac{1}{4\,m\,\omega }\,\big(\,2\,\mathbf{p}_{13}(0)\,\sin{\omega\,t}  +  \mathbf{p}_{13}^+(0)\,\sin{2\,\omega\,t}\,\big)
\\ &
\mathbf{r}_2(t) =   \frac{1}{2}\,\big(\, \mathbf{r}_{24}(0)\,\cos{\omega\,t}   +  \mathbf{r}_{24}^+(0)\,\cos{2\,\omega\,t}\,\big)
 \, + \,   \frac{1}{4\,m\,\omega }\,\big(\,2\,\mathbf{p}_{24}(0)\,\sin{\omega\,t}  +  \mathbf{p}_{24}^+(0)\,\sin{2\,\omega\,t}\,\big)
 \\ &
 \mathbf{r}_3(t) \ = \ \mathbf{r}_1(t \pm 2\,\tau) \qquad ; \qquad 
 \mathbf{r}_4(t) \ = \ \mathbf{r}_2(t \pm 2\,\tau) \ ,
\end{aligned} 
\end{equation}
with time delay $\tau = \frac{\pi}{2\,\omega}$, and
\begin{equation}
\mathbf{r}_{ij}^+(t) \ \equiv  \ \mathbf{r}_{i}(t) \, + \, \mathbf{r}_{j}(t) \quad , \qquad  \mathbf{p}_{ij}^+(t) \ \equiv \ \mathbf{p}_{i}(t) \ + \ \mathbf{p}_{j}(t) \ ,
\end{equation}

\noindent here individual quantities $\mathbf{r}_{i}(0)$, $\mathbf{p}_{i}(0)$, and relative ones $\mathbf{r}_{ij}(0)$ and $\mathbf{p}_{ij}(0)$ correspond to the initial conditions at $t=0$. 
In the center-of-mass frame, the relations $\mathbf{r}_{13}^+(t) = -\mathbf{r}_{24}^+(t)$ and $\mathbf{p}_{13}^+(t) = -\mathbf{p}_{24}^+(t)$ hold. 
\end{proof}

\noindent Effectively, the phase space associated with the relative motion is 12-dimensional. Moreover, the time evolution of the relative vectors $\mathbf{r}_{13}$ and $\mathbf{r}_{24}$ reads
\begin{equation}
\begin{aligned}
& \mathbf{r}_{13}(t) \ = \   \mathbf{r}_{13}(0)\,\cos{\omega\,t} \ + \ \frac{1}{m\,\omega}\,\mathbf{p}_{13}(0)\,\sin{\omega\,t} 
\\ &
\mathbf{r}_{24}(t) \ = \   \mathbf{r}_{24}(0)\,\cos{\omega\,t} \ + \ \frac{1}{m\,\omega}\,\mathbf{p}_{24}(0)\,\sin{\omega\,t} \ .
\label{r13r24}
 \end{aligned}
\end{equation}

\noindent Hence, for an arbitrary set of initial conditions, the motion exhibits two isomorphic 2-body choreographies, as evinced by  (\ref{trajec}). In this typical scenario, with the linear forces allowing both attraction and repulsion, particle pair $1,3$ and pair $2,4$ each follow their own periodic orbit. These two orbits have the same general functional form but usually different parameters (hence, different shapes), see Fig.\ref{F1}. In this generic case, both relative distances $r_{13}(t)$ and $r_{24}(t)$ vary in time (neither remains constant).

\vspace{0.15cm}

\noindent We emphasize that neither of the two pairs forming a 2-body choreography can decay into more elementary motions, such as particles 1 and 3 following distinct trajectories. Such a decay is prohibited—protected by symmetry. Consequently, these configurations are inherently stable.

\vspace{0.15cm}

\noindent Depending on the initial conditions, {we can distinguish several classes of orbits}:

\begin{itemize}
        
\item \textbf{(I)} One 2-body pair rigid, one 2-body pair non-rigid: Suppose the initial conditions satisfy
\begin{equation}
    \label{iniconC0}
 {p}_{13}(0)  \ = \ m\,\omega\,{r}_{13}(0) \qquad ; \qquad  \mathbf{r}_{13}(0) \,\perp\, \mathbf{p}_{13}(0) \ .
    \end{equation}
In this case, $r_{13}(t)$ stays constant for all $t$, equal to its initial value (see (\ref{r13r24})), while the other pair’s separation $r_{24}(t)$ is still time-dependent (non-constant). By the symmetry ${\cal S}_2(13)\otimes {\cal S}_2(24)$, the opposite situation is also possible, i.e. $r_{24}(t)$ remains constant and $r_{13}(t)$ varies.

\vspace{0.4cm}

\item \textbf{(II)} Two rigid 2-body choreographies (distinct radii): If we impose the initial conditions (\ref{iniconC0}) on both pairs simultaneously, namely
    \begin{equation}
    \label{iniconC1}
    \begin{aligned}
      &  {p}_{13}(0)  \ = \ m\,\omega\,{r}_{13}(0) \qquad ; \qquad  \mathbf{r}_{13}(0) \,\perp\, \mathbf{p}_{13}(0) 
    \\ &
        {p}_{24}(0)  \ = \ m\,\omega\,{r}_{24}(0) \qquad ; \qquad  \mathbf{r}_{24}(0) \,\perp\,\mathbf{p}_{24}(0)\ ,
        \end{aligned}
    \end{equation}
then both $r_{13}$ and $r_{24}$ remain constant in time: 
\[
r_{13}(t)\  = \  r_{13}(0)  \qquad , \qquad  r_{24}(t) \ = \ r_{24}(0) \ ,
\]
    
\noindent for all $t$. In this situation, each two-body motion is \textit{rigid} (illustrated in Fig.\ref{F2}), although generally the radii differ (so $r_{13}(t) \neq r_{24}(t)$ in this case).

\vspace{0.2cm}   

\item \textbf{(III)} In this class of trajectories, the existence of four-body choreographic solutions, with the two pairs' motions intertwined, is established in the following Lemma:
\begin{lemma}\label{lem1}
\textit{The necessary and sufficient conditions to obtain a four-body choreography from the general solutions (\ref{trajec}) are:}
\begin{equation}
m\,\omega\,\mathbf{r}_{13}(0)\ =\ \pm\,\mathbf{p}_{24}(0),\qquad m\,\omega\,\mathbf{r}_{24}(0)\ =\ \mp\,\mathbf{p}_{13}(0)\ .
\label{con4brt}
\end{equation}
\end{lemma}
\begin{proof}
The four-body choreographic motion requires the condition  
\begin{equation}
\mathbf{r}_{2}(t) \ = \  \mathbf{r}_{1}(t + \tau), \qquad \tau = \frac{\pi}{2\,\omega}.
\end{equation}
To verify this, we evaluate $\mathbf{r}_{1}(t + \tau)$ using equation~(\ref{trajec}), which yields:
\begin{equation}
\mathbf{r}_{2}(t) = \frac{1}{2} \left(-\mathbf{r}_{13}(0)\, \sin{\omega t} - \mathbf{r}_{13}^+(0)\, \cos{2\omega t} \right) \\
+ \frac{1}{4m\omega} \left(2\,\mathbf{p}_{13}(0)\, \cos{\omega t} - \mathbf{p}_{13}^+(0)\, \sin{2\omega t} \right)\ , \nonumber
\end{equation}
to be compared with the explicit form of $\mathbf{r}_2(t)$ given in equation~(\ref{trajec}).  
Since the equality must hold for all times $t$, each orthogonal trigonometric component must match independently. The components involving $\cos(2\omega t)$ and $\sin(2\omega t)$ automatically coincide due to the center-of-mass condition. The remaining two terms yield the following constraints:
\begin{subequations}
\begin{align}
& \mathbf{r}_{24}(0) \  = \  \frac{\mathbf{p}_{13}(0)}{m\,\omega}, \qquad \text{(cosine balance)} 
\label{eq:bal4-cost}
\\ &
\frac{\mathbf{p}_{24}(0)}{m\,\omega} \  = \  -\mathbf{r}_{13}(0), \qquad \text{(sine balance)} 
\label{eq:bal4-sint}
\end{align}
\end{subequations}

The remaining particles $3$ and $4$ follow similarly time-shifted trajectories:
\begin{equation}
\mathbf{r}_{3}(t) = \mathbf{r}_{1}(t + 2\tau), \qquad \mathbf{r}_{4}(t) = \mathbf{r}_{2}(t + 2\,\tau) = \mathbf{r}_{1}(t + 3\,\tau),
\end{equation}
as defined by equation~(\ref{trajec}). Together, these satisfy the temporal symmetry required to complete the four-body choreographic motion. The $\pm$ signs in (\ref{con4brt}) arise from time-reversal invariance.
\end{proof}
Again, both $r_{13}(t)$ and $r_{24}(t)$ vary periodically in a coordinated way (neither distance is constant), as shown in  Fig.\ref{F4}.  Eventually, from the 12 \textit{free} initial conditions in the space of relative motion, the 4-body choreography, due to lemma
1, has a maximum of eight possibly different constants that can be
independently set. The parametric representation of the common path followed
by the four bodies is
\begin{multline}
\mathbf{r}\left(t\right)=\frac{1}{2}\big(\mathbf{r}_{13}\!\left(0\right)\cos\left(\omega t\right)+\mathbf{r}_{13}^{+}\!\left(0\right)\cos\left(2\omega t\right)\big)\\
+\frac{1}{2\,m\,\omega}\left(\mathbf{p}_{13}\!\left(0\right)\sin\left(\omega t\right)+\frac{1}{2}\mathbf{p}_{13}^{+}\!\left(0\right)\sin\left(2\omega t\right)\right)\ .\label{eq:4body_orbit}
\end{multline}

The 4-body choreography is also stable. Small deviations in the initial conditions (keeping the center of mass at rest) lead to a fragmentation into two 2-body choreographies, each following trajectories that closely resemble the original, unperturbed 4-body motion. { It is also worth mentioning that the
particular values of the coefficients in the potential (\ref{Vp}), play a fundamental role in determining whether choreographic motion can arise and persist.\\}

\vspace{0.4cm}

\item \textbf{(IV)} Symmetric four-body choreography: Finally, by further imposing the conditions (\ref{iniconC1}) on the previous case (III)
ensures that the relative distances  $r_{13}(t),\, r_{24}(t)$ become both constant and equal. For instance, by selecting the initial values
{\small
\begin{equation}
    \label{iniconC3}
    \begin{aligned}
   &   
   \mathbf{r}_{1}(0) \,= \,(1,\,0) \quad ,\quad \mathbf{r}_{2}(0) \,= \,\frac{1}{2}(-1,\,1)\quad ,\quad \mathbf{r}_{3}(0) \,= \,(0,0) \quad ,\quad \mathbf{r}_{4}(0) \,= \,-\frac{1}{2}(1,\,1)
   \\ &
   \mathbf{p}_{1}(0) \,= \,(0,\,\frac{3\,m\,\omega}2) \quad ,\quad \mathbf{p}_{2}(0) \,= \,-m\,\omega\,(\frac{1}{2},\,1)\quad ,\quad \mathbf{p}_{3}(0) \,= \,(0,\,\frac{m\,\omega}2) 
   \\ &
   \mathbf{p}_{4}(0) \,= \,m\,\omega\,(\frac{1}{2},\,-1) \ ,
    \end{aligned}
    \end{equation}
}

    or equivalently,

{\small
\begin{equation}
    \label{iniconC4}
    \begin{aligned}
   &   
   \mathbf{r}_{13}(0) \,= \,(1,\,0) \quad ,\quad \mathbf{r}_{24}(0) \,= \,(0,\,1)\quad ,\quad \mathbf{r}^+_{13}(0) \,= \,(1,0) \quad ,\quad \mathbf{r}^+_{24}(0) \,= \,(-1,\,0)
   \\ &
   \mathbf{p}_{13}(0) \,= \,m\,\omega\,(0,\,1) \quad ,\quad \mathbf{p}_{24}(0) \,= \,m\,\omega\,(-1,\,0)\quad ,\quad \mathbf{p}^+_{13}(0) \,= \,m\,\omega\,(0,\,2) 
   \\ &
   \mathbf{p}^+_{24}(0) \,= \,m\,\omega\,(0,\,-2) \ ,
    \end{aligned}
    \end{equation}
}
\vspace{0.2cm}

the 4-body trisectrix limaçon choreography~\cite{Fernandez-Guasti2025} is obtained
\begin{equation}
\label{trilimax}
\mathbf{r}(t)\ = \  \frac{1}{2}\big(\cos(\omega t) + \cos(2\omega t)\big)\,\hat{\mathbf{e}}_{x} + \frac{1}{2}\big(\sin(\omega t) + \sin(2\omega t)\big)\,\hat{\mathbf{e}}_{y}\ .
\end{equation}
In this case $r_{13}(t)\, =\, r_{24}(t)=1$, see Fig. \ref{F3}.

\vspace{0.2cm}

\item \textbf{(V)} General four-body choreography. The solution when the four bodies follow the same trajectory is characterized by eight independent constants. The orbits can be very different depending on the values of these parameters. Some of the orbits resemble trifolium roses, lemniscates and other less symmetric curves. {They} can be constructed through the superposition of particular solutions, owing to the linearity of the equations of motion. Starting from the trisectrix limaçon solution (\ref{trilimax}), scaling the trigonometric components individually leads to four degrees of freedom.  Altogether, these can be reparametrized as independent coefficients in the form
\begin{equation}
\mathbf{r}(t) = \frac{1}{2}\big(a\cos(\omega t) + b\cos(2\omega t)\big)\,\hat{\mathbf{e}}_{x} + \frac{1}{2}\big(c\sin(\omega t) + d\sin(2\omega t)\big)\,\hat{\mathbf{e}}_{y}.
\end{equation}
A distinct, but equally valid, solution is obtained by exchanging the trigonometric functions associated with the basis vectors
\begin{equation}
\mathbf{r}(t) = \frac{1}{2}\big(c'\sin(\omega t) + d'\sin(2\omega t)\big)\,\hat{\mathbf{e}}_{x} + \frac{1}{2}\big(a'\cos(\omega t) + b'\cos(2\omega t)\big)\,\hat{\mathbf{e}}_{y}\ .
\end{equation}
The general solution emerges from the linear superposition of these two expressions, thus yielding a full set of eight independent parameters. Strictly speaking, one could also include time-reversed solutions via the symmetry transformation \(t \to -t\); however, their superposition leads to counter-propagating trajectories that generically result in particle collisions.

\end{itemize}

These results can be further
generalized in two ways, on the one hand, an $n>4$ number of bodies
can be considered. On the other, the coefficients in the arguments of the trigonometric functions, here specified as 1 and 2 can be generalized to $\Gamma_{x}$ and $\Gamma_{y}$ in $\mathbb{Z}$. These integer values ensure that the orbits will be periodic. The force strength coefficients $\kappa_{ij}$, defined by the harmonic
type potential $V\left(r_{ij}\right)=\frac{1}{2}\sum\kappa_{ij}r_{ij}^{2}$,
establish the values of $\Gamma_{x},\Gamma_{y}$. The problem can also be tackled
the other way around, where $\kappa_{ij}$ coefficients are adjusted
depending on the given $\Gamma_{x},\Gamma_{y}$ constants. Not all $n,\Gamma_{x},\Gamma_{y}$ possibilities, together with the initial positions and momenta constants, satisfy Hamilton's equations of motion. Furthermore, even if the equations of motion are satisfied, collisions may occur in some of these configurations. Nonetheless, there is a vast set of values where choreographies can be  obtained.

\begin{figure}[h]
\begin{center}
\includegraphics[width=10.5cm]{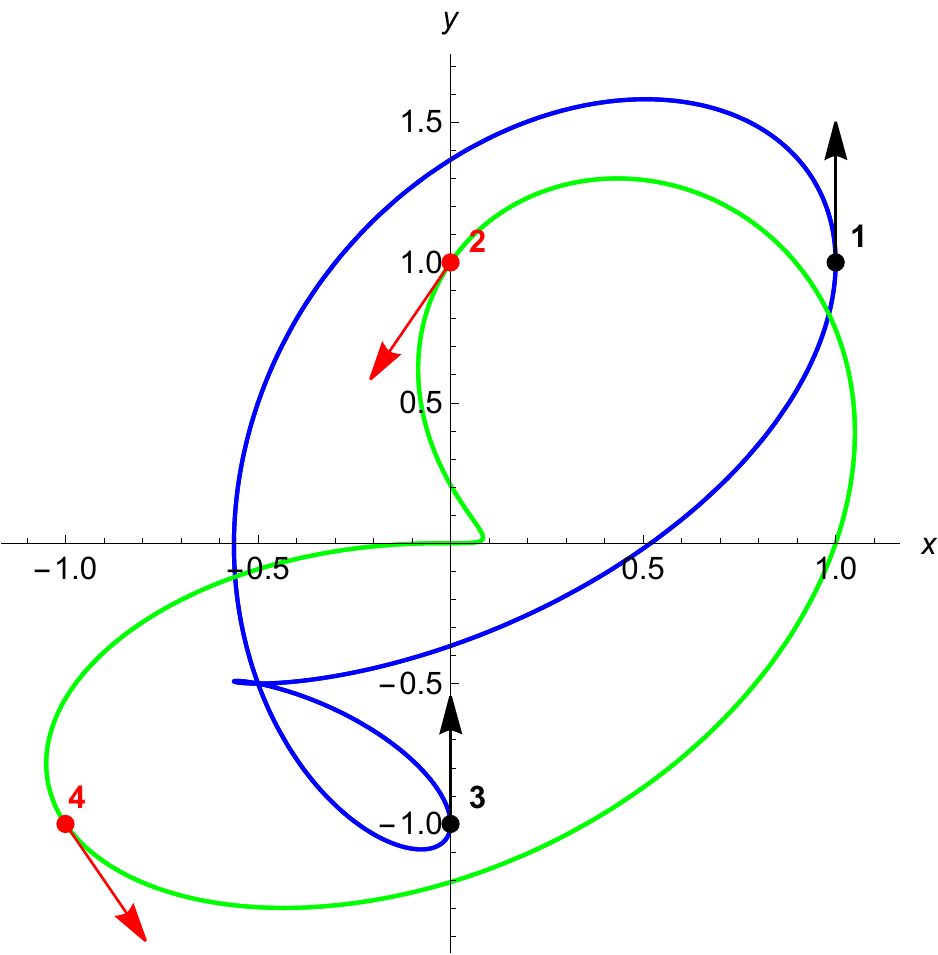}
	\caption{\textbf{Generic case}. Two 2-body choreographies with non-constant relative distances $r_{13}(t)$ and $r_{24}(t)$. Particles 1 and 3 (shown in black) share a common trajectory (blue curve), while particles 2 and 4 (shown in red) follow a distinct but identical path (green curve).
    Here $m=1$, $\omega=1$ and the initial conditions are given by ${\mathbf r}_1=(1,1)$, ${\mathbf r}_2=(0,1)$, ${\mathbf r}_3=(0,-1)$, ${\mathbf r}_4=(-1,-1)$, ${\mathbf p}_1=(0,1.5)$, ${\mathbf p}_2=(-0.5,-1)$, ${\mathbf p}_3=(0,0.5)$, ${\mathbf p}_4=(0.5,-1)$. The arrows representing the momentum vectors indicate direction only; their lengths are not to scale and are intended solely as a guide to the eye.}
    \label{F1}
\end{center}
\end{figure}  


\begin{figure}[h]
\begin{center}
\includegraphics[width=11.5cm]{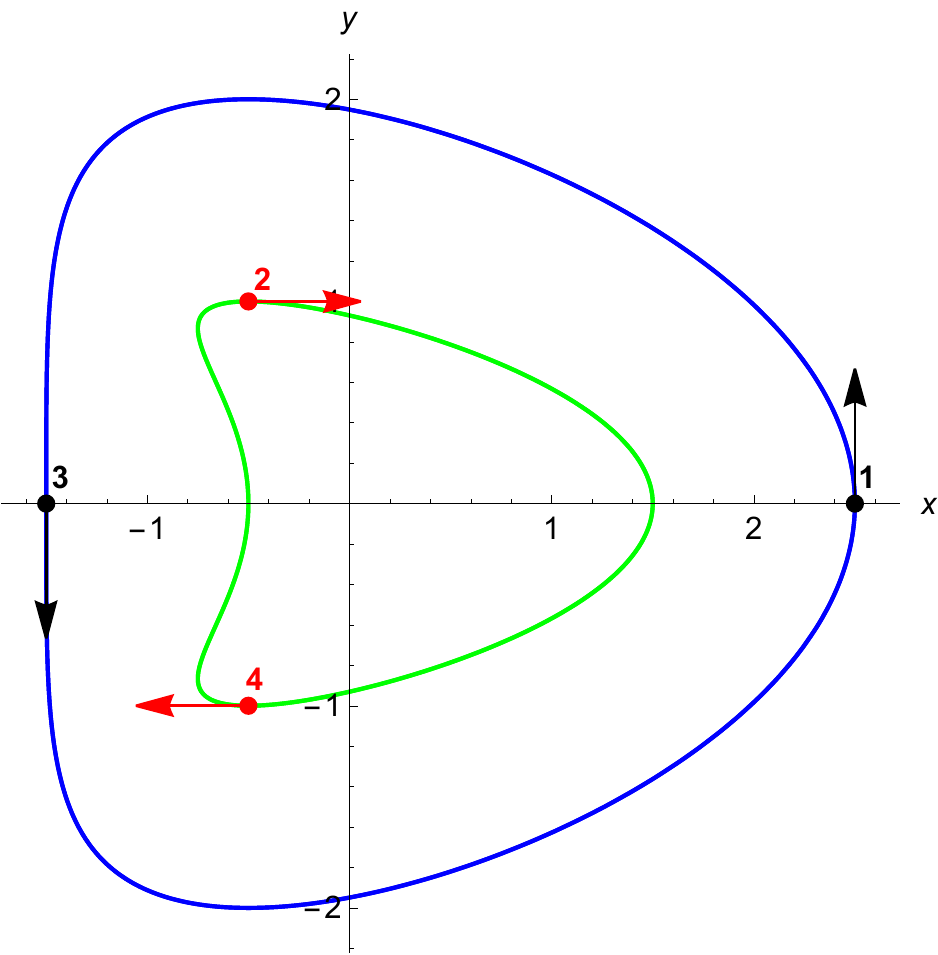}
\caption{\textbf{Class (II)}. Two rigid 2-body choreographies (distinct radii): two distinct two-body choreographies are presented, characterized by constant but different relative distances: $r_{13}(t)=4$ and $r_{24}(t)=2$. Particles 1 and 3 (shown in black) share a common trajectory (blue curve), while particles 2 and 4 (shown in red) follow a distinct but identical path (green curve). Here $m=1$, $\omega=1$ and the initial conditions are given by ${\mathbf r}_1=(2.5,0)$, ${\mathbf r}_2=(-0.5,1)$, ${\mathbf r}_3=(-1.5,0)$, ${\mathbf r}_4=(-0.5,-1)$, ${\mathbf p}_1=(0,2)$, ${\mathbf p}_2=(1,0)$, ${\mathbf p}_3=(0,-2)$, ${\mathbf p}_4=(-1,0)$. Momentum arrows indicate direction only; their scale is not accurate and serves illustrative purposes.}
\label{F2}
\end{center}
\end{figure}

\begin{figure}[h]
\begin{center}
\includegraphics[width=7cm]{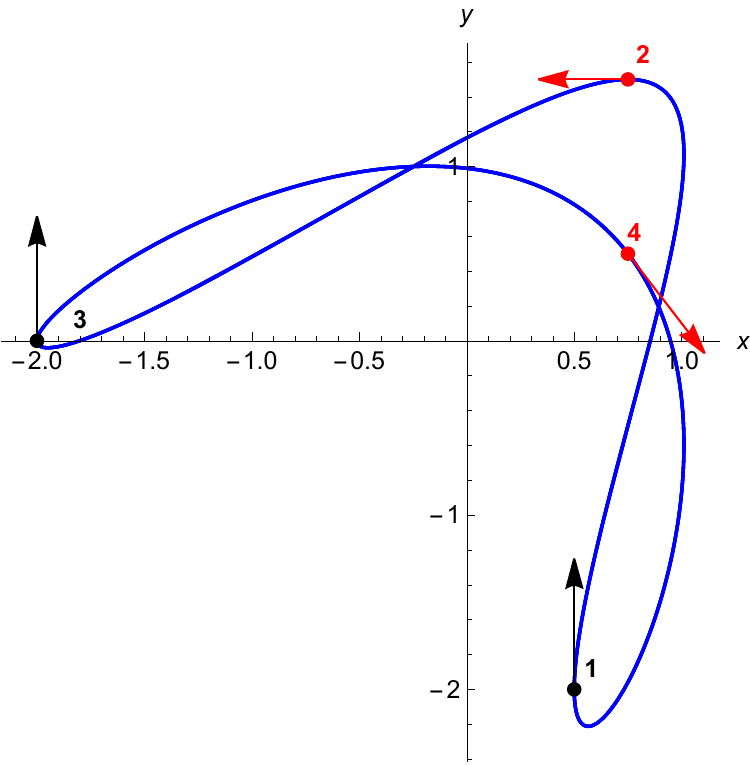}
\caption{\small \textbf{Class (III)}. Generic four-body choreography: all four particles move along the same path (blue curve). Here $m=1$, $\omega=1$ and the initial conditions are given by ${\mathbf r}_1=(0.5,-2)$, ${\mathbf r}_3=(-2,0)$, ${\mathbf p}_1=(0,1.5)$, ${\mathbf p}_3=(0,0.5)$. The remaining initial conditions were obtained from (\ref{con4brt}). The arrows representing the momentum vectors indicate direction only; their lengths are not to scale and are intended solely as a guide to the eye.}
\label{F4}
\end{center}
\end{figure}
\begin{figure}[h]
\begin{center}
\includegraphics[width=7cm]{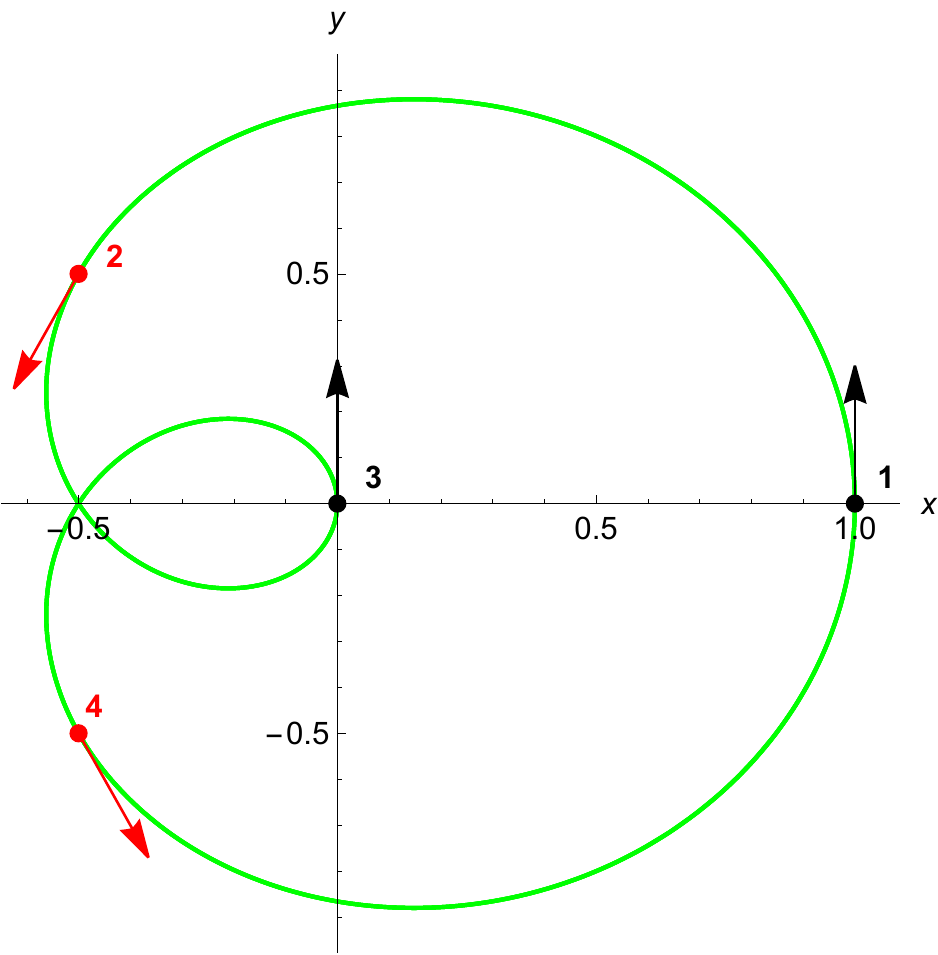}
\caption{\small \textbf{Class (IV)}. Symmetric four-body choreography: Four-body trisectrix limaçon choreography with constant and equal relative distances $r_{13}(t)=r_{24}(t)=1$. Here $m=1$, $\omega=1$, the initial conditions taken from (\ref{iniconC3}). Momentum arrows indicate direction only; their scale is not accurate and serves illustrative purposes.} 
\label{F3}
\end{center}
\end{figure}

\clearpage

\section{Complete separation of variables}\label{sec3}

In this section, we return to the original Hamiltonian $\cal H$ \eqref{Ha} and show that, in full generality, it admits a complete separation of variables. To this end, introduce in (\ref{Ha}) vectorial two-dimensional Jacobi-like coordinates:

\begin{equation}
\begin{aligned}
\label{JacC}
& {\mathbf J}_0 \  = \ \frac{1}{4}(\, {\mathbf r}_1 \ + \ {\mathbf r}_2 \ + \  {\mathbf r}_3 \ + \ {\mathbf r}_4 \,)
\\ &
{\mathbf J}_1 \ = \ \frac{1}{\sqrt{2}}(\, {\mathbf r}_2 \ - \ {\mathbf r}_1 \,)
\\ &
{\mathbf J}_2 \ = \ \sqrt{\frac{2}{3}}(\, {\mathbf r}_3 \ - \ \frac{1}{2}({\mathbf r}_1\ + \  {\mathbf r}_2 ) \,)
\\ &
{\mathbf J}_3 \ = \ \frac{\sqrt{3}}{2}(\, {\mathbf r}_4 \ - \ \frac{1}{3}({\mathbf r}_1\ + \  {\mathbf r}_2\ + \  {\mathbf r}_3 ) \,) \ .
\end{aligned}    
\end{equation}
In these variables (\ref{JacC}), the original Hamiltonian (\ref{Ha}) reads
\begin{equation}
\label{HaJ}
{\cal H} \ = \ \frac{1}{8\,m}\bigg(\, \mathbf{P}_{J_0}^2 \,+\,4\,\mathbf{P}_{J_1}^2 \,+\,4\,\mathbf{P}_{J_2}^2 \,+\,4\,\mathbf{P}_{J_3}^2 \,\bigg) \ + \ V(J) \ ,
\end{equation}
with potential 
\begin{equation}
\label{VpJ}
    V(J) \ = \ m\,\omega^2\,\bigg(\, \frac{5}{4} J_1^2   
\ + \ \frac{3}{4} J_2^2  \ + \ J_3^2  \ - \ \frac{\sqrt{3}}{2}(\,{\mathbf J}_1 \cdot {\mathbf J}_2  ) \ + \ \sqrt{\frac{3}{2}}(\,{\mathbf J}_1 \cdot {\mathbf J}_3  )  \ - \ \frac{1}{\sqrt{2}}(\,{\mathbf J}_2 \cdot {\mathbf J}_3  )  \,\bigg) \ ,
\end{equation}
where $J_i=|{\mathbf J}_i|$ denotes the magnitude of ${\mathbf J}_i$, and $\mathbf{P}_{J_i}$ is the canonical conjugate momentum associated with $\mathbf{J}_i$, $i=0,1,2,3$. The center-of-mass variable $\mathbf{J}_0$ is cyclic and therefore entirely decouples from the internal relative dynamics.\\

\subsubsection{Reduced Hamiltonian in relative space}

\noindent Hereafter, we solely focus on the motion in the relative space, i.e. $\mathbf{P}_{J_0}=0$ and $\mathbf{J}_0=0$ which define the aforementioned center-of-mass frame. The corresponding reduced \textit{relative} Hamiltonian takes the form
\begin{equation}
\label{Hrel}
{\cal H}_{\rm rel} \ = \ {\cal H}\mid_{_{\mathbf{ P}_0=0}} \ = \ \frac{1}{2\,m}\bigg(\,\mathbf{P}_{J_1}^2 \,+\,\mathbf{P}_{J_2}^2 \,+\,\mathbf{P}_{J_3}^2 \,\bigg) \ + \ V(J) \ .
\end{equation}

\noindent Now, we change variables $(\mathbf{J}_1,\,\mathbf{J}_2,\,\mathbf{J}_3\,) \rightarrow (\mathbf{Q}_1,\,\mathbf{Q}_2,\,\mathbf{Q}_3\,)$ using a $(5\,\pi/6)$-rotation around ${\mathbf J}_3$
\begin{equation}
\begin{aligned}
& 
{\mathbf J}_1 \ = \ \cos(5\,\pi/6)\,{\mathbf Q}_1  \ - \ \sin(5\,\pi/6)\,{\mathbf Q}_2 \ = \ -\frac{\sqrt{3}}{2}\,{\mathbf Q}_1  \ - \ \frac{1}{2}\,{\mathbf Q}_2
\\ &
{\mathbf J}_2 \ = \ \sin(5\,\pi/6)\,{\mathbf Q}_1  \ + \ \cos(5\,\pi/6)\,{\mathbf Q}_2  \ = \ \frac{1}{2}\,{\mathbf Q}_1  \ - \ \frac{\sqrt{3}}{2}\,{\mathbf Q}_2
\\ &
{\mathbf J}_3 \ = \ {\mathbf Q}_3  \ .
\end{aligned}    
\end{equation}
In these new $Q-$coordinates, we have
\begin{equation}
\label{HrelQ}
{\cal H}_{\rm rel}  \ = \ \frac{1}{2\,m}\bigg(\,\mathbf{P}_{Q_1}^2 \,+\,\mathbf{P}_{Q_2}^2 \,+\,\mathbf{P}^3_{Q_3} \,\bigg) \ + \ V(Q) \ ,
\end{equation}
here
\begin{equation}
\label{VpQ}
    V(Q) \ = \ m\,\omega^2\bigg(\, \frac{3}{2} Q_1^2   
\ + \ \frac{1}{2} \,Q_2^2  \ + \  \,Q_3^2  \ - \ \sqrt{2}\,(\,{\mathbf Q}_1 \cdot {\mathbf Q}_3  )   \,\bigg) \ ,
\end{equation}
($Q_i=|{\mathbf Q}_i|$). Thus, the variable ${\mathbf Q}_2$ is separated out. Finally, performing a further rotation $(\mathbf{Q}_1,\,\mathbf{Q}_2,\,\mathbf{Q}_3,\,) \rightarrow (\mathbf{U}_1,\,\mathbf{U}_2,\,\mathbf{U}_3,\,)$ around ${\mathbf Q}_2$
\begin{equation}
\begin{aligned}
& 
{\mathbf Q}_1 \ = \ \cos\big(\tan ^{-1}(\sqrt{2})\,\big)\,{\mathbf U}_1  \ - \ \sin\big(\tan ^{-1}(\sqrt{2})\,\big)\,{\mathbf U}_3  \ = \ \frac{1}{\sqrt{3}}\,{\mathbf U}_1  \ - \ \sqrt{\frac{2}{3}}\,{\mathbf U}_3
\\ &
{\mathbf Q}_2 \ = \ {\mathbf U}_2
\\ &
{\mathbf Q}_3 \ = \ \sin\big(\tan ^{-1}(\sqrt{2})\,\big)\,{\mathbf U}_1  \ + \ \cos\big(\tan ^{-1}(\sqrt{2})\,\big)\,{\mathbf U}_3 \ = \ \sqrt{\frac{2}{3}}\,{\mathbf U}_1  \ + \  \frac{1}{\sqrt{3}}\,{\mathbf U}_3\ ,
\end{aligned}    
\end{equation}
we obtain
\begin{equation}
\label{HrelU}
{\cal H}_{\rm rel}  \ = \ \frac{1}{2\,m}\bigg(\,\mathbf{P}_{U_1}^2 \,+\,\mathbf{P}_{U_2}^2 \,+\,\mathbf{P}^2_{U_3} \,\bigg) \ + \ \frac{1}{2}\,m\,\,\bigg(\, \Omega_1^2 \,U_1^2   
\ + \ \Omega_2^2 \,U_2^2  \ + \ \Omega_3^2\,U_3^2     \,\bigg) \ ,
\end{equation}
with angular frequencies
\begin{equation}
\label{fre}
\Omega_1\ = \  \omega \ ,\qquad \Omega_2 \ = \ \omega \ ,\qquad \Omega_3\ = \  2\,\omega\,.
\end{equation}
Therefore, the reduced Hamiltonian (\ref{Hrel}) governing the relative motion with six degrees of freedom decomposes into the sum of three independent harmonic oscillators with two degrees of freedom each and natural frequencies $\Omega_1= \omega,\,\Omega_2= \omega,\,\Omega_3= 2\,\omega,\,$ respectively. This is in complete agreement with the angular frequencies appearing in (\ref{trajec}). {In summary, the original system \eqref{Ha} admits a complete separation of variables using Jacobi-like coordinates given by
\begin{equation}
\begin{aligned}
\label{finalT}
& {\mathbf U}_0 \ = \ \tfrac{1}{4}\,({\mathbf r}^+_{13}\,+\,{\mathbf r}^+_{24}) \qquad ;
&&\quad \mathbf{P}_{U_0} \ = \  {\mathbf p}^+_{13}\,+\,{\mathbf p}^+_{24}
\\[4pt]
& {\mathbf U}_1 \ = \ -\tfrac{1}{\sqrt{2}}\,{\mathbf r}_{24} \qquad ;
&&\quad \mathbf{P}_{U_1} \ = \ -\tfrac{1}{\sqrt{2}}\, \mathbf{p}_{24} 
\\[4pt]
& {\mathbf U}_2 \ = \ \tfrac{1}{\sqrt{2}}\,{\mathbf r}_{13}\qquad ;
&&\quad \mathbf{P}_{U_2} \ = \ \tfrac{1}{\sqrt{2}} \,\mathbf{p}_{13}
\\[4pt]
& {\mathbf U}_3 \ = \ \tfrac{1}{2}\,({\mathbf r}^+_{24}\,-\,{\mathbf r}^+_{13}\,)\qquad ;
&&\quad \mathbf{P}_{U_3} \ = \ \tfrac{1}{2}\,({\mathbf p}^+_{24}\,-\,{\mathbf p}^+_{13}\,) \ ,
\end{aligned}
\end{equation}
where ${\mathbf U}_0={\mathbf J}_0$ is the center-of-mass variable. In the center-of-mass frame, the canonical transformation (\ref{finalT}) simplifies to 
\begin{equation}
\begin{aligned}
\label{UvarCenter}
& {\mathbf U}_0 \ = \ 0 \qquad ;
&&\quad \mathbf{P}_{U_0} \ = \ 0
\\[4pt]
& {\mathbf U}_1 \ = \ -\tfrac{1}{\sqrt{2}}\,{\mathbf r}_{24} \qquad ;
&&\quad \mathbf{P}_{U_1} \ = \ -\tfrac{1}{\sqrt{2}}\, \mathbf{p}_{24} 
\\[4pt]
& {\mathbf U}_2 \ = \ \tfrac{1}{\sqrt{2}}\,{\mathbf r}_{13}\qquad ;
&&\quad \mathbf{P}_{U_2} \ = \ \tfrac{1}{\sqrt{2}} \,\mathbf{p}_{13}
\\[4pt]
& {\mathbf U}_3 \ = \ {\mathbf r}^+_{24} \ = \ -{\mathbf r}^+_{13}   \qquad ;
&&\quad \mathbf{P}_{U_3} \ = \ {\mathbf p}^+_{24} \ = \ -{\mathbf p}^+_{13} \ .
\end{aligned}
\end{equation}

Motivated by the discrete permutation symmetry \( \mathcal{S}_2(13) \otimes \mathcal{S}_2(24) \) in (\ref{Ha}), the choice of the \(\mathbf{U}\)-variables appears natural. In this section, we rigorously demonstrate that they indeed allow for a complete separation of variables.\\
}

\textbf{Remark}: \textit{The original Hamiltonian~\eqref{Ha} can be effectively reduced to a sum of three independent two-dimensional isotropic harmonic oscillators and a free particle in two dimensions. More importantly, this decomposition offers a clear conceptual framework (superintegrability of anisotropic multidimensional harmonic oscillator) for exploring generalizations to the planar $(D = 2)$ $n$-body problem with $n > 4$, as well as to higher-dimensional settings $D > 2$.} 

\subsection{Relative motion: maximal superintegrability}
\label{relmotionsup}

From the original system \eqref{Ha}, upon separating the center of mass motion, we are left with a completely separable Hamiltonian $\mathcal{H}_{\rm rel} $ with six degrees of freedom (\ref{HrelU}). In this section, we demonstrate that $\mathcal{H}_{\rm rel}$ is, in fact, maximally superintegrable \cite{Miller2013, grigoriev2018superintegrable, Marquette2017}.

\noindent As derived in (\ref{HrelU}), the post-separation Hamiltonian $H_{\rm rel}$ can be viewed as the sum of two uncoupled systems: a four-dimensional isotropic harmonic oscillator $\mathcal{H}_{4D}$ with frequency $\omega$ and a two-dimensional harmonic oscillator $\mathcal{H}_{2D}$ with frequency $2\omega$ (as shown in (\ref{fre})\,). In other words,

\begin{equation}
\mathcal{H}_{\rm rel} \ = \ \mathcal{H}_{4D}\  + \ \mathcal{H}_{2D}\ .
\end{equation}

This direct sum structure implies a large number of conserved quantities. In fact, each subsystem is maximally superintegrable on its own, so together they contribute a total of ten independent integrals of motion (with six of them in mutual involution), confirming that $H_{\rm rel}$ is (at least) minimally superintegrable.

\vspace{0.1cm}

More concretely, the $2D$ subsystem $\mathcal{H}_{2D}$ alone provides three algebraically independent integrals. For example, in Cartesian coordinates $\mathbf{U}_3 \equiv (u^{(3)}_x, u^{(3)}_y)$ and $\mathbf{P}_{U_3} \equiv (P^{(3)}_{x}, P^{(3)}_{y})$, one can verify that (i) the $2D$ Hamiltonian $\mathcal{H}_{2D}$, (ii) the angular momentum $L$ and (iii) the $x$-direction energy $ H^{(3)}_x$

\begin{equation}
\label{InU3}
L\  =\ u^{(3)}_x\, P^{(3)}_{y} \ -\  u^{(3)}_y\, P^{(3)}_{x}, \quad \text{and} \quad H^{(3)}_x \ = \ \frac{{(P^{(3)}_{x})}^2}{2\,m}\ + \  \frac{1}{2}\,m\, \Omega_3^2\, {(u^{(3)}_x})^2 \ ,    
\end{equation}
($\Omega_3=2\,\omega$) are conserved under time evolution, they obey $\{L, \mathcal{H}_{\text{rel}}\} = 0$ and $\{H^{(3)}_x, \mathcal{H}_{\text{rel}}\} = 0$, where $\mathcal{H}_{\text{rel}}$ is the reduced Hamiltonian~(\ref{HrelU}) in the space of relative motion. On the other hand, the Hamiltonian \(\mathcal{H}_{4D}\) contributes with seven well-known additional conserved quantities, detailed in Appendix~\ref{secA}.

Using (\ref{UvarCenter}) we can express the above integral $L$ in terms of the original Cartesian individual variables $\mathbf{r}_i=(x_i,\,y_i)$,\,$\mathbf{p}_i=(p_{x_i},\,p_{y_i})$, $i=1,2,3,4$. Explicitly,
\begin{equation}
    L \ = \ (x_2 \,+\, x_4)(p_{y_2} \,+\,p_{y_4})\ -\  (y_2 \,+\, y_4)(p_{x_2} \,+\,p_{x_4}) \ , 
\end{equation}
and similarly for $H^{(3)}_x$. 

\vspace{0.1cm}

Furthermore, the commensurability of the system’s natural frequencies, given by $\frac{\Omega_3}{\Omega_{1(2)}} = 2$, leads to the emergence of an additional integral of motion. This extra eleventh integral can be derived from the Holt potential \cite{Holt}, as described in [\cite{Bonatsos}, Eqs. (37)–(38)], thereby establishing the maximal superintegrability of the system. Specifically, consider the Cartesian coordinates $\mathbf{U}_1 \equiv (u^{(1)}_x, u^{(1)}_y)$ and the corresponding momenta $\mathbf{P}_{U_1} \equiv (P^{(1)}_x, P^{(1)}_y)$. It is then straightforward to verify that
\begin{equation}
  {\cal I} \ = \   {(P^{(1)}_{x})}^2\,P^{(3)}_{y} \ + \ 4\,m^2\,\omega^2\,u^{(1)}_x\,u^{(3)}_y\,P^{(1)}_{x} \ - \ m^2\,\omega^2\,{(u^{(1)}_{x})}^2\,P^{(3)}_{y} \ ,
\end{equation}
Poisson commutes with $\mathcal{H}_{\text{rel}}$, thus, $\frac{d}{dt}{\cal I}=0$.

According to the celebrated Nekhoroshev theorem \cite{nekhoroshev1972action}, the existence of eleven integrals of motion (the maximum possible number for this system) explains why, in the center-of-mass frame (bounded motion), all trajectories are periodic, regardless of the initial conditions.

It is worth emphasizing that the maximal superintegrability of the system alone does not account for the existence of the 4-body choreography. This exceptional configuration is more accurately described by the notions of \textit{particular integral} and \textit{particular involution}~\cite{Escobar-Ruiz2024}, see below.

\section{The four-body trisectrix lima\c{c}on choreography: \textit{particular integrals} and \textit{particular involution} }\label{sec4}

\subsection{Particular integrability}\label{}

Let $H$ be a Hamiltonian system with \( n \) degrees of freedom. We adopt the notion of particular integrals as introduced in \cite{Turbiner_2013, Escobar-Ruiz2024}, see also \cite{Azuaje2025}. Let $H$ be the Hamiltonian of the system and $f:\Gamma\to \mathbb{R}$ a smooth function on the phase space $\Gamma$. We call $f$ a \textit{particular integral} of $H$ if
\[
 \{ f, H \} \ = \  a\, f ,
\] 
for some smooth function $a$ on $\Gamma$. In the special case $a=0$, $f$ is a standard first integral (in the Liouville sense), meaning $f$ is constant along every trajectory (we will also refer to this as a global integral of motion). In general, if $a \neq 0$, $f$ is not conserved for arbitrary solutions; however, there may exist special trajectories on which $f$ stays constant. Likewise, two functions $f, g \in C^\infty(\Gamma)$ are said to be in \textit{particular involution} if $\{ f, g \} \neq 0$ in general, but vanishes for all time along some specific trajectory $\gamma$. In other words, $f$ and $g$ do not commute under the Poisson bracket universally, yet they do Poisson commute when the system is restricted to a certain invariant trajectory (or submanifold) $\gamma$.

\subsection{Existence of \textit{particular integrals} in the four-body trisectrix lima\c{c}on choreography}

\noindent The specific initial conditions (\ref{iniconC3}) for the trisectrix lima\c{c}on choreography enforce strong constraints: specifically, under those conditions the magnitudes of the vectors $\mathbf {U}_1(t), \mathbf {U}_2(t), \mathbf {U}_3(t)$ defined in (\ref{UvarCenter}) (which are proportional to the relative distances $r_{13}(t), r_{24}(t)$, and $r^+_{13}(t)=r^+_{24}(t)$, respectively) remain constant in time, as do the magnitudes of their conjugate momenta $P_{U_i}(t)$. Geometrically, this means the four-body  motion on the trisectrix lima\c{c}on is confined to three rotational degrees of freedom (three fixed angles), making the system effectively behave as three decoupled angular oscillators. In this scenario, additional conserved quantities emerge. In fact, the scalar products $\mathbf{r}_{13} \cdot \mathbf{r}_{24}$ and $\mathbf{p}_{13} \cdot \mathbf{p}_{24}$ themselves act as \textit{particular integrals} of motion: their non-zero Poisson brackets with $H_{\rm rel}$ vanish on the choreography $\gamma_{\rm trisectrix}$:

\begin{align}
\{\mathbf{r}_{13} \cdot \mathbf{r}_{24}\,,\, \mathcal{H}_{\text{rel}}\}\big|_{\gamma_{\rm trisectrix}} = \{\mathbf{p}_{13} \cdot \mathbf{p}_{24}\,,\, \mathcal{H}_{\text{rel}}\}\big|_{\gamma_{\rm trisectrix}} = 0 \ .
\end{align}

\textbf{Remark:} \textit{The quantities $\mathbf{r}_{13} \cdot \mathbf{r}_{24}$ and $\mathbf{p}_{13} \cdot \mathbf{p}_{24}$ are particular integrals that serve to characterize the entire family of four-body trajectories belonging to class (IV).}

\vspace{0.2cm}

For instance, the most generic trajectory (\ref{trajec}) of this system is characterized by the global integral

\begin{equation}
\label{Igen}
   {\cal I} \ = \  \mathbf{r}_{13} \cdot \mathbf{r}_{24} \ + \ \frac{1}{m^2\,\omega^2}\,\mathbf{p}_{13} \cdot \mathbf{p}_{24} \ .
\end{equation}

In contrast, in the \textit{rigid} choreography case (when both relative distances $r_{13}$ and $r_{24}$ are constant; case (IV) in our classification of initial conditions), the two terms in ${\cal I}$ are separately conserved (each term individually remains constant). For the generic motion, by comparison, $\mathbf{r}_{13}\cdot \mathbf{r}_{24}$ and $\mathbf{p}_{13}\cdot \mathbf{p}_{24}$ are only conserved in tandem – their weighted sum $I$ is constant along the motion, but each part oscillates.

Furthermore, the two global integrals $L$ and $H^{(3)}_x$ obtained above (\ref{InU3}) are not in involution, since their Poisson bracket does not vanish. Indeed, one finds 
\begin{align}
    \{L, H^{(3)}_x\} \ = \  \frac{1}{m} P^{(3)}_{x} \,P^{(3)}_{y} \ + \  m \,\Omega_3^2 \,u^{(3)}_x\, u^{(3)}_y \ .
\end{align}
However, for the special 4-body trisectrix lima\c{c}on choreography under study, which arises from the tuned initial conditions (\ref{iniconC3}), these two integrals become involutive along that trajectory $\gamma_{\rm trisectrix}$. In other words, although $ \{L, H^{(3)}_x\} \neq0$ in general, we have:
\begin{align}
\label{LHx}
    \{L, H^{(3)}_x\}\big|_{\gamma_{\rm trisectrix}} = 0 \ .
\end{align}
Hence, in this case the pair of integrals $(L,\, H^{(3)}_x)$ are in \textit{particular involution}.

\section{{Choreographic} fragmentation} 
\label{sec5}

Assume that the initial conditions given in equations~(\ref{eq:bal4-cost})–(\ref{eq:bal4-sint}) are satisfied; these define a generic four-body choreography of class (III). In this configuration, the particular integral~(\ref{Igen}) vanishes identically. However, a slight perturbation of the system causes~(\ref{Igen}) to lose its conserved character, leading to a dynamical transition wherein the original motion fragments into two independent two-body choreographies. We refer to this phenomenon as choreographic fragmentation. A supplementary motion file is included to visualize this fragmentation.

\subsection{Fragmentation of 4-Body Choreographies}

{Consider the initial conditions (\ref{con4brt}) described in Lemma~\ref{lem1}, which define the necessary and sufficient conditions to obtain a four-body choreography from
the general solutions (\ref{trajec})}, i.e., a motion in which all four particles follow the same orbit with a uniform time shift. At a specific time $t_{\mathrm{ex}}$, the momenta of some particles---specifically $\mathbf{p}_{13}(t_{\mathrm{ex}})$ and/or $\mathbf{p}_{24}(t_{\mathrm{ex}})$---are abruptly modified. This momentum alteration breaks the symmetry conditions~\eqref{eq:bal4-cost}--\eqref{eq:bal4-sint}, thereby causing the original 4-body choreography to split into two independent 2-body choreographies. 

We refer to these instantaneous momentum alterations as \textit{boosts}, drawing an analogy with Lorentz boosts in relativistic mechanics. In this context, the dynamics are assumed to be continuous in position but discontinuous in momentum at the instant of the boost; that is, the spatial configuration remains unchanged immediately before and after $t_{\mathrm{ex}}$, while the momenta are modified.

Let this boost occur at time $t_{\mathrm{ex}} = \frac{2\pi}{\omega} \beta$, with $\beta \in \mathbb{Z}$, meaning the choreography undergoes an abrupt change after completing $\beta$ full cycles. The momentum adjustment must preserve the center-of-mass condition, requiring that the total linear momentum remains zero, specifically $\mathbf{p}_{13}^{+} = -\mathbf{p}_{24}^{+}$.

To illustrate, suppose particle 2 experiences a momentum increment $\boldsymbol{\delta}$, while particle 3 receives the opposite change:
\[
\mathbf{p}_2(t_{\mathrm{ex}}) = \mathbf{p}_2(0) + \boldsymbol{\delta}, \quad \mathbf{p}_3(t_{\mathrm{ex}}) = \mathbf{p}_3(0) - \boldsymbol{\delta}.
\]
Consequently, the relative momenta evolve as:
\[
\mathbf{p}_{24}(t_{\mathrm{ex}}) = \mathbf{p}_{24}(0) + \boldsymbol{\delta}, \quad
\mathbf{p}_{13}(t_{\mathrm{ex}}) = \mathbf{p}_{13}(0) + \boldsymbol{\delta}, \quad
\mathbf{p}_{13}^{+}(t_{\mathrm{ex}}) = \mathbf{p}_{13}^{+}(0) - \boldsymbol{\delta}.
\]
The post-boost trajectories correspond to two decoupled 2-body choreographies:
\begin{multline*}
\mathbf{r}_{1,3}(t) = \frac{1}{2} \left( \pm \mathbf{r}_{13}(t_{\mathrm{ex}}) \cos(\omega t) + \mathbf{r}_{13}^{+}(t_{\mathrm{ex}}) \cos(2\omega t) \right) \\
+ \frac{1}{2\,m\,\omega} \left( \pm \left[\mathbf{p}_{13}(t_{\mathrm{ex}}) + \boldsymbol{\delta}\right] \sin(\omega t) + \frac{1}{2} \left[\mathbf{p}_{13}^{+}(t_{\mathrm{ex}}) - \boldsymbol{\delta}\right] \sin(2\omega t) \right),
\end{multline*}
\begin{multline*}
\mathbf{r}_{2,4}(t) = \frac{1}{2} \left( \pm \mathbf{r}_{24}(t_{\mathrm{ex}}) \cos(\omega t) - \mathbf{r}_{13}^{+}(t_{\mathrm{ex}}) \cos(2\omega t) \right) \\
+ \frac{1}{2\,m\,\omega} \left( \pm \left[\mathbf{p}_{24}(t_{\mathrm{ex}}) + \boldsymbol{\delta}\right] \sin(\omega t) - \frac{1}{2} \left[\mathbf{p}_{13}^{+}(t_{\mathrm{ex}}) - \boldsymbol{\delta}\right] \sin(2\omega t) \right) \ ,
\end{multline*}
where the upper sign $(+)$ corresponds to $\mathbf{r}_{1}(\mathbf{r}_{2})$ and the lower sign $(-)$ to $\mathbf{r}_{3}(\mathbf{r}_{4})$. From a mathematical standpoint, the boost $\boldsymbol{\delta}$ can point in any direction within the plane, provided it is counterbalanced by an opposite momentum in the system, thus preserving the total momentum. However, it is physically insightful to model a possible mechanism for this perturbation. 

Consider the following scenario: at time $t_{\mathrm{ex}}$, bodies 2 and 3 exchange equal masses $\mu_{\mathrm{ex}}$, each ejecting mass at velocity $\mathbf{v}_{\mathrm{ex}}$. The recoil momentum per body is $\frac{\boldsymbol{\delta}}{2} = \mu_{\mathrm{ex}} \mathbf{v}_{\mathrm{ex}}$. Upon receiving the ejected mass from its counterpart, each body experiences an additional momentum transfer of $\frac{\boldsymbol{\delta}}{2}$, resulting in a net change of $\pm \boldsymbol{\delta}$. This boost is directed along the line joining the two particles.

At $t_{\mathrm{ex}}$, the positions coincide with those at $t = 0$, i.e., $\mathbf{r}_j(t_{\mathrm{ex}}) = \mathbf{r}_j(0)$. For instance,
\[
\mathbf{r}_2(t_{\mathrm{ex}}) = \mathbf{r}_1\left(t_{\mathrm{ex}} + \frac{\pi}{2\omega}\right) = \frac{1}{2}(-\hat{\mathbf{e}}_x + \hat{\mathbf{e}}_y), \quad \text{and} \quad \mathbf{r}_3(t_{\mathrm{ex}}) = 0.
\]
The slope of the line joining particles 2 and 3 is $-\pi/4$, and the boost vector takes the form $\boldsymbol{\delta} = \delta(-\hat{\mathbf{e}}_x + \hat{\mathbf{e}}_y)$. In Fig.~\ref{fig:boost23}, this is illustrated with a specific case: 
\[
\boldsymbol{\delta} = \mu_{\mathrm{ex}} \cdot \frac{3}{4} \cdot 10\, \omega\, (-\hat{\mathbf{e}}_x + \hat{\mathbf{e}}_y), \quad \mu_{\mathrm{ex}} = \frac{m}{10}.
\]
This leads to the following momentum transitions:
\[
\mathbf{p}_2(0) = m\,\omega \left(-\tfrac{5}{4} \hat{\mathbf{e}}_x - \tfrac{1}{4} \hat{\mathbf{e}}_y\right), \quad
\mathbf{p}_3(0) = -m\,\omega \left(\hat{\mathbf{e}}_x + \tfrac{1}{2} \hat{\mathbf{e}}_y\right).
\]
Meanwhile, the momenta of bodies 1 and 4 remain unchanged and tangent to both the original and resulting trajectories, as they are not affected by the boost.

\begin{figure}[h]
\centering
\includegraphics[scale=0.4]{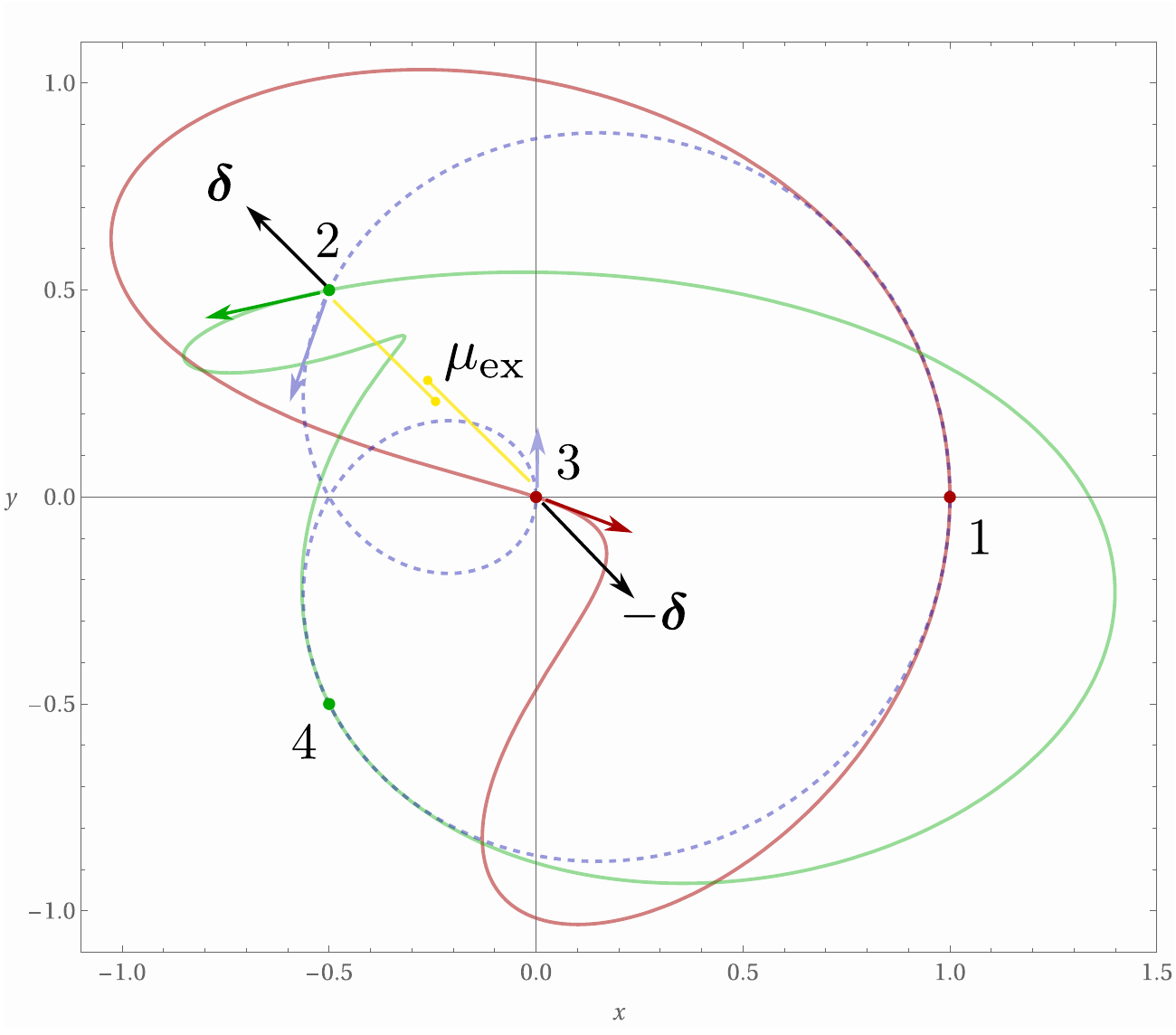}
\caption{\label{fig:boost23}{Choreographic fragmentation. Four bodies initially follow a 4-body trisectrix limaçon orbit (dotted blue). After $\beta$ full cycles, at time $t_{\mathrm{ex}} = \beta\, \omega\, t$, a pair of opposite boosts is applied to bodies 2 and 3 (represented by the yellow equal mass exchange $\mu_{\mathrm{ex}}$). The choreography fragments into two 2-body motions: bodies 1 and 3 follow the red orbit, and bodies 2 and 4 the green orbit. In this example, $m=1$, $\omega=1$ and $\boldsymbol{\delta} = \tfrac{3}{4}(-\hat{\mathbf{e}}_x + \hat{\mathbf{e}}_y)$. An animation file of this fragmentation is presented in a supplementary file.} }
\end{figure}

\subsection{Fusion of Two 2-Body Choreographies}

We now consider the inverse process. Suppose we start from two independent 2-body choreographies governed by the trajectories:
\begin{subequations}
\begin{align}
\mathbf{r}_{1,3}(t) &= \frac{1}{2} \left( \pm \mathbf{r}_{13}(0) \cos(\omega t) + \mathbf{r}_{13}^{+}(0) \cos(2\omega t) \right) \nonumber \\
&\quad + \frac{1}{2\,m\,\omega} \left( \pm \mathbf{p}_{13}(0) \sin(\omega t) + \frac{1}{2} \mathbf{p}_{13}^{+}(0) \sin(2\omega t) \right), \\
\mathbf{r}_{2,4}(t) &= \frac{1}{2} \left( \pm \mathbf{r}_{24}(0) \cos(\omega t) - \mathbf{r}_{13}^{+}(0) \cos(2\omega t) \right) \nonumber \\
&\quad + \frac{1}{2\,m\,\omega} \left( \pm \mathbf{p}_{24}(0) \sin(\omega t) - \frac{1}{2} \mathbf{p}_{13}^{+}(0) \sin(2\omega t) \right),
\end{align}
\end{subequations}
where the choreography conditions are initially not satisfied, i.e., 
\[
\mathbf{r}_{24}(0) \neq \frac{1}{m\,\omega} \mathbf{p}_{13}(0), \quad \text{or} \quad \frac{1}{m\,\omega} \mathbf{p}_{24}(0) \neq -\mathbf{r}_{13}(0).
\]
To enable fusion into a 4-body choreography, two pairs of opposite momentum boosts are applied at time $t_{\rm fus} = \frac{2\pi}{\omega} \beta$. The required boosts are:
\[
\boldsymbol{\delta}_{\rm fus\,13} = \frac{1}{2} \left( m\,\omega\, \mathbf{r}_{24}(0) - \mathbf{p}_{13}(0) \right), \quad
\boldsymbol{\delta}_{\rm fus\,24} = \frac{1}{2} \left( m\,\omega\, \mathbf{r}_{13}(0) - \mathbf{p}_{24}(0) \right).
\]
Once these specific boosts are applied (ensuring the center-of-mass constraints $\mathbf{r}_{13}^{+} = -\mathbf{r}_{24}^{+}$ and $\mathbf{p}_{13}^{+} = -\mathbf{p}_{24}^{+}$) the two independent 2-body systems are transformed into a unified 4-body choreography.

Unlike fragmentation, which can occur under a wide range of perturbations, fusion into a 4-body choreography demands finely tuned momentum adjustments, akin to executing a precise propulsion maneuver in spacecraft dynamics.

\section{Conclusions}\label{sec6}

In this work, we investigated a planar four-body system with a quadratic pairwise potential and demonstrated that its reduced (center-of-mass–free) six-degree-of-freedom dynamics is maximally superintegrable. In particular, we showed that the system can be separated into three independent harmonic oscillators with commensurate frequencies $(\omega, \omega, 2\omega)$, yielding eleven independent integrals of motion (five more than the six required for Liouville integrability).  \\

\noindent A central result of our study is the analysis of the analytical four-body choreography along a trisectrix limaçon. We found that maximal superintegrability alone is not sufficient to produce this choreographic motion. Instead, the choreography arises due to the existence of \textit{particular integrals} that are conserved only on this specific trajectory.  This explains the exceptional nature of the four-body limaçon choreography within the broader family of periodic solutions. \\

Furthermore, we described an interesting fragmentation phenomenon: under certain conditions the four-body choreography splits, or “bifurcates,” into two isomorphic two-body choreographies. We characterized this choreographic fragmentation by showing that it coincides with the breakdown of the particular integrals that sustain the four-body motion. This represents a new type of dynamical transition in a superintegrable system, wherein a highly symmetric motion (the four-body choreography) fractures into two simpler subsystems when the fine-tuned conditions for the particular integrals are not met. The inverse process, involving the fusion of two-body choreographies into a four-body configuration, was also described.\\

Our findings open several promising avenues for future research. In particular, one could pursue a systematic classification of analogous choreographic solutions in other superintegrable $n$-body systems, exploring whether similar particular integrals govern their dynamics and what forms of choreographic fragmentation may arise.  \\

Finally, exploring quantum analogues of these phenomena would be highly interesting. Quantizing the four-body limaçon system or related models could reveal whether the classical constants of motion and choreographic fragmentation have counterparts in the quantum regime (e.g. degenerate energy levels or conserved quantum numbers reflecting the superintegrability). Such studies would deepen our understanding of how geometric and algebraic structures in classical many-body dynamics translate into quantum systems.\\

In summary, this work elucidates how a unique four-body choreography emerges from a maximally superintegrable framework supplemented by particular integrals of motion. The choreographic fragmentation we identified underscores a new mechanism by which symmetry and integrability can give rise to complex collective behavior. We believe these insights motivate further theoretical developments, helping to connect the geometric nature of choreographies with the algebra of integrals of motion, and guiding future explorations into $n$-body dynamics and their quantum counterparts.

\section*{Acknowledgments}

The author A.M. Escobar-Ruiz thanks Rafael Azuaje for his interest in this work and for the helpful discussions. {The authors gratefully acknowledge the valuable remarks provided by the anonymous referees.} A.M. Escobar-Ruiz would like to thank the support from Consejo Nacional de Humanidades, Ciencias y Tecnologías (CONAHCyT) of Mexico under Grant CF-2023-I-1496 and from UAM research grant DAI 2024-CPIR-0. MFG
acknowledges support from CONAHCyT, project CF-2023-I-1864.

\backmatter

\bigskip

\begin{appendices}

\section{Integrals of motion for a 4D isotropic harmonic oscillator}\label{secA}

By the direct sum of ${\mathbf U}_1 \equiv (q_1,q_2)$ and ${\mathbf U}_2 \equiv (q_3,q_4)$ we form a $4-$dimensional vector ${\mathbf U} = {\mathbf U}_1 \oplus {\mathbf U}_2 \equiv (q_1,q_2,q_3,q_4)$, see (\ref{HrelU}). This 4-dimensional isotropic harmonic oscillator is governed by the Hamiltonian:
\begin{equation}
    H_{4D} = \sum_{i=1}^{4} \left( \frac{p_i^2}{2\,m} \ + \  \frac{1}{2}\, m \,\omega^2 q_i^2 \right),
\end{equation}
where $q_i$ and $p_i$ are the conjugate generalized coordinates and momenta, $m$ is the mass of each particle, and $\omega$ is the common frequency of oscillation.

This system exhibits a rich algebraic structure due to its symmetries, which include both rotational symmetries and hidden symmetries arising from its superintegrability \cite{Evans}. The system has multiple conserved quantities: the total energy $H_{4D}$ is an integral of motion, the six angular momentum components:
\begin{equation}
    L_{ij} \ = \ q_i\, p_j \ - \ q_j\, p_i, \quad (1 \leq i < j \leq 4),
\end{equation}
which satisfy the Poisson brackets:
\begin{equation}
    \{ L_{ij}, L_{kl} \} = \delta_{ik} L_{jl} + \delta_{jl} L_{ik} - \delta_{il} L_{jk} - \delta_{jk} L_{il} \ ,
\end{equation}
are conserved. They obey the $\mathfrak{so}(4)$ algebra. Less obvious integrals of motion come from the Fradkin tensor \cite{Fradkin}:
\begin{equation}
    I_{ij} = \frac{p_i\, p_j}{m} \ + \  m\, \omega^2 \,q_i\, q_j,
\end{equation}
which commute under the Poisson brackets $ \{I_{ij}, I_{kl} \}=0$. This means that $I_{ij}$ represent a set of mutually commuting integrals of motion that provide a hidden symmetry not immediately evident from the rotational properties alone.
The trace of $I_{ij}$ is related to the Hamiltonian,
\begin{equation}
    \sum_{i=1}^{4} I_{ii} = \frac{2\,H_{4D}}{m}\ ,
\end{equation}
thus, only three of the integrals $I_{ii}$ are algebraically independent.

The mixed brackets between $I_{ij}$ and $L_{kl}$ read:
    \begin{equation}
        \{ I_{ij}, L_{kl} \} = \delta_{ik} I_{jl} + \delta_{jl} I_{ik} - \delta_{il} I_{jk} - \delta_{jk} I_{il} \ .
    \end{equation}
This structure indicates that $I_{ij}$ forms a rank-2 symmetric tensor under the rotation group $SO(4)$.
The full dynamical symmetry group (16 generators in total) of the system is $U(4)$. Although there are several conserved quantities, only 7 are algebraically independent. This ensures integrability and confirms the system as maximally superintegrable. For instance, we take the Hamiltonian $H_{4D}$, three components of $L_{ij}$ (say $L_{12}$,$L_{13}$ and $L_{14}$), and three elements of $I_{ij}$ (say $I_{11},I_{22},I_{33}$). 

The following table summarizes the Poisson brackets between the algebraically independent integrals of motion:

\begin{table}[h]
    \centering
    \renewcommand{\arraystretch}{1.5}
    \begin{tabular}{|c|c|c|}
        \hline
        \textbf{Poisson Brackets} & \textbf{Result} & \textbf{Description} \\
        \hline
        $\{ L_{12}, L_{13} \}$ & $L_{23}$ & $SO(4)$ angular momentum algebra \\
        $\{ L_{12}, L_{14} \}$ & $L_{24}$ & $SO(4)$ angular momentum algebra \\
        $\{ L_{13}, L_{14} \}$ & $L_{34}$ & $SO(4)$ angular momentum algebra \\
        $\{ I_{ii}, I_{jj} \}$ & $0$ & Commuting quadratic integrals \\
        $\{ I_{ii}, L_{ij} \}$ & $2 \,\omega^2 \,I_{ij}$ & Fradkin tensor coupling to $SO(4)$ \\
        $\{ H_{4D}, L_{ij} \}$ & $0$ & Energy is a Casimir element \\
        $\{ H_{4D}, I_{ii} \}$ & $0$ & Energy commutes with quadratic integrals \\
        \hline
    \end{tabular}
    \caption{Poisson algebra of the 7 algebraically independent integrals of motion ($H_{4D},L_{12}, L_{13},L_{14}, I_{11}, I_{22}, I_{33}$) in the 4D isotropic harmonic oscillator.}
    \label{tab:poisson_algebra}
\end{table}

\end{appendices}

\bibliography{sn-bibliography}

\end{document}